\documentclass[journal, twocolumn]{IEEEtran}

\usepackage{amsmath,amssymb,graphicx,epsfig,fancyhdr,amsthm,tabulary,epstopdf,stmaryrd}
\usepackage{hyperref}
\usepackage{cite}
\usepackage{multirow}
\usepackage{float}
\usepackage{amsfonts}
\usepackage{subfigure}
\usepackage{footnote}
\makesavenoteenv{tabular}
\usepackage{tabulary,bigstrut,array,multirow,booktabs}%
\usepackage{tikz}
\usetikzlibrary{calc}
\usepackage{pdfsync}
\usepackage{bm}

\usepackage{multirow, cite}
\usepackage{xcolor,colortbl, soul}
\usepackage{tabularx,booktabs}

\definecolor{Gray}{gray}{0.85}
\definecolor{LightCyan}{rgb}{0.88,1,1}
\definecolor{LightBlue}{rgb}{0.75,0.936,1.00}
\sethlcolor{LightBlue}
\newcolumntype{a}{>{\columncolor{Gray}}c}
\newcolumntype{b}{>{\columncolor{white}}c}

\newtheorem{thm}{Theorem}[]
\newtheorem{cor}{Corollary}
\newtheorem{lem}{Lemma}

\theoremstyle{remark}
\newtheorem{rem}{Remark}

\theoremstyle{definition}

\begin{document}

\title{ Optimal Beamforming for Gaussian MIMO Wiretap Channels with Two Transmit Antennas}
\author{\IEEEauthorblockN{Mojtaba Vaezi,  Wonjae Shin, and H. Vincent Poor}\\
}

\maketitle


\begin{abstract}
A Gaussian multiple-input multiple-output
\textit{wiretap} channel in which the eavesdropper and legitimate receiver are
equipped with arbitrary numbers of antennas and the transmitter has two antennas is
studied in this paper. Under an average  power constraint, the optimal input covariance to obtain
the secrecy capacity of this channel is unknown, in general. In this paper,
the input covariance matrix required  to achieve the capacity is determined.
It is shown that the secrecy capacity of this channel can be achieved by \textit{linear precoding}.
The optimal precoding and power allocation schemes that maximize the achievable secrecy rate,
and thus achieve the capacity,  are developed subsequently.
 The secrecy capacity is then compared with the achievable secrecy rate of
  \textit{generalized singular value decomposition} (GSVD)-based
  precoding, which is the best previously proposed technique for this problem.
 Numerical results demonstrate that substantial
gain can be obtained in secrecy rate between  the proposed  and GSVD-based precodings.
\end{abstract}

\section{Introduction}
{\let\thefootnote\relax\footnotetext{

Manuscript received January 16, 2017; revised May 5, 2017, and accepted July 16, 2017.
This research was supported in part by the U. S. National Science
Foundation under Grant CMMI-1435778, and  in part by a Canadian NSERC fellowship.
This paper was partly presented at IEEE International Symposium on Information Theory  (ISIT), in Aachen, 2017 \cite{vaezi2017isit}.

Mojtaba Vaezi and H. Vincent Poor are with the Department of Electrical Engineering, Princeton University, Princeton, NJ, USA (e-mail:
\{mvaezi,poor\}@princeton.edu). Wonjae Shin is with Department of Electrical and Computer Engineering, Seoul National University, Seoul, Korea (e-mail:wonjae.shin@snu.ac.kr).

Digital Object Identifier

}}

Wireless networks have become an indispensable part of our daily life and security/privacy of information transfer
via these networks is crucial.
Unfortunately, wireless communication systems are inherently
insecure due to the broadcast nature of the medium. Hence, wireless security has been
an important concern for many years.
Traditionally, security is  provided at the upper layers 
of wireless networks via \textit{cryptographic} techniques, wherein the legitimate user
has a secret key to decode its message.
Security can be
also offered at the lowest layer (physical layer), e.g.,
via  beamforming or artificial noise  injection \cite{mukherjee2014principles}, to support and supplement
existing cryptographic protocols.

Physical layer security has attracted widespread  attention as a means of augmenting wireless security \cite{mukherjee2014principles}.
Physical layer security is based on the information theoretic secrecy that can be provided by
 physical communication channels, an idea that was first proposed by Wyner \cite{wyner1975wire},
 in the context of the {\it wiretap} channel. In this channel,  a transmitter wishes to
transmit information to a \textit{legitimate} receiver while keeping
the information secure from an \textit{eavesdropper}.  Wyner demonstrated that it is possible to
have both \textit{reliable} and \textit{secure} communication between the transmitter and legitimate receiver
in the presence of an eavesdropper  under certain circumstances.
The basic principal is that  the channel of the legitimate receiver   should be stronger in some sense than that of the
eavesdropper.

With the rapid advancement of  multi-antenna techniques, security enhancement
in multiple-input multiple-output (MIMO) wiretap channels, see Fig.~\ref{fig:sysmodel},
has drawn significant attention.  A big step toward understanding the MIMO Gaussian
wiretap channel was taken in \cite{khisti2010secure,oggier2011secrecy,liu2009note}
where a closed-form expression for the capacity  of this channel
was established.  However, to compute this  expression,  the  input covariance
matrix that maximizes it needs to be determined. Under an average power constraint, such a
matrix is unknown in general.\footnote{Under a power-covariance constraint, the capacity
expression and corresponding covariance matrix is found in \cite{liu2009note} and \cite{bustin2009mmse}, respectively.}
Recently, numerical solutions have been proposed to compute a transmit covariance matrix for this
channel \cite{li2013transmit,steinwandt2014secrecy,loyka2015algorithm}. These numerical approaches  solve
the underlying non-convex optimization problem iteratively. Despite their efficiency,
there is still motivation to find an analytical
solution for this problem and study  simpler techniques for secure communication, e.g., based on linear precoding.

\begin{figure}[t]
\centering
\includegraphics[scale=.42]{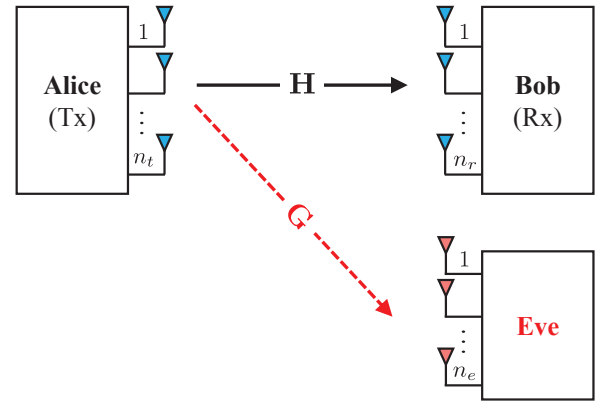}
\caption{MIMO Gaussian wiretap channel with $n_t$, $n_r$, and $n_e$  antennas,  at the transmitter (Alice), legitimate receiver (Bob), and eavesdropper (Eve).
}
\label{fig:sysmodel}
\end{figure}

Precoding  is a technique for exploiting transmit diversity via weighting the information stream.
\textit{Singular value decomposition} (SVD) precoding with \textit{water-filling} power allocation is a well-known example that achieves the  capacity of the MIMO channel.
 Khisti and  Wornell \cite{khisti2010secure} proposed a \textit{generalized SVD} (GSVD)-based precoding scheme with equal power allocation  for the MIMO Gaussian wiretap channel.
The optimal power allocation scheme for  GSVD precoding in the MIMO Gaussian wiretap channel was obtained in \cite{fakoorian2012optimal}.
Although GSVD precoding gets  close to the capacity in certain antenna configurations,
it is neither capacity-achieving nor very close to capacity, in general.
Despite its importance and years of research, optimal transmit/receive
strategies to  maximize the secure rate in MIMO wiretap channels remain unknown, in general.
Linear beamforming transmission has, however, been proved to be optimal for the  special case  of  $n_t=2$, $n_r=2$, and $n_e = 1$
in \cite{shafiee2009towards}.
 It is also known to be the optimal communication
strategy for multiple-input single-output (MISO) wiretap channels  \cite{shafiee2007achievable, khisti2010secureMISOME}.

Recently, a closed-form solution for
the optimal covariance matrix has been  found  when the channel is strictly degraded
and another condition on the channel matrices, which is equivalent to a lower threshold on the
transmitted power, holds \cite{loyka2012optimal, fakoorian2013full, loyka2016optimal}.
The combination of this result and the unit-rank solution of \cite{khisti2010secureMISOME}
can give the optimal covariance matrix for the case of two transmit antennas \cite{loyka2016optimal}.
The optimal solution is, however, still open in general.

In this paper,  we characterize  optimal precoding and power allocation   for
MIMO Gaussian wiretap channels in which the legitimate receiver and
eavesdropper have arbitrary numbers of antennas but the transmitter has two antennas.
This proves that linear beamforming transmission can be optimal  for a much broader class of MIMO Gaussian wiretap
channels.  
 Our approach in finding the optimal covariance matrix is completely different from that of \cite{fakoorian2013full} and \cite{loyka2016optimal}.
 It does not require the degradedness condition and thus provides the optimal solution for both
 full-rank and rank-deficient cases in one shot.
The  proposed beamforming and power allocation schemes result in a  computable capacity with a reasonably low complexity.
It requires  searching over  two scalars (power allocation) at most.
In addition, the proposed beamforming and power allocation schemes
can bring notably high gain over GSVD-based beamforming, as
confirmed by simulation results.

It is worth highlighting  that the new precoding and power allocation techniques are  applicable
to and optimal for MIMO  channels without secrecy, simply by setting the eavesdroppers channel to zero.
In such cases, power allocation is even simpler and does not require a search.

Secure transmission strategies  in multi-antenna networks with various constraints and/or in different settings, e.g.,
with energy-efficiency \cite{zhang2014energy}, finite memory \cite{shlezinger2017secrecy},
joint source-relay precoding \cite{wang2014joint},
game-theoretic precoding\cite{fang2015game}, and varying eavesdropper channel states \cite{he2014mimo} have been
considered recently.

The rest of the paper is organized as follows. In Section~\ref{sec:model},
we describe the system model. In Section~\ref{sec:main},
we  reformulate the secrecy rate problem and  propose  linear precoding
and   power allocation schemes to achieve the capacity of the MIMO/MISO wiretap channels.
 In Section~\ref{sec:ext}, we show that the proposed precoding and power allocation schemes are also optimal for
 MIMO/MISO channels without an eavesdropper and we discuss possible extensions of the proposed precoding method. We
  present numerical results in Section~\ref{sec:sim} before concluding the paper in Section~\ref{sec:con}.

Throughout this work, we use notations $\mathrm{tr}(\cdot)$, $\det(\cdot)$, $(\cdot)^t$, and
$(\cdot)^H$ to denote the trace,  determinant, transpose, and conjugate transpose of a matrix,
respectively. Matrices are written in bold capital letters and
vectors are written in bold small letters. $\mathbf{A}\succeq \mathbf{0} $ means that $\mathbf{A}$ is a positive semidefinite matrix, and
$\mathbf{I}_m$ represents the identity matrix of size $m$.

\section{System Model and Preliminaries}
\label{sec:model}
Consider a MIMO Gaussian wiretap channel, in which
a transmitter (Alice) wishes to  communicate with a legitimate receiver (Bob) in the presence of an eavesdropper (Eve),
as shown in Fig.~\ref{fig:sysmodel}. The nodes are equipped with $n_t$, $n_r$, and $n_e$  antennas, respectively.
Let $\mathbf{H} \in \mathbb{R}^{n_r\times n_t}$ and
$\mathbf{G} \in \mathbb{R}^{n_e\times n_t}$ be the channel matrices for the legitimate user and  eavesdropper. 
Both channels are assumed
to undergo independent and identically distributed (i.i.d.) Rayleigh fading,
where the channel gains are  real Gaussian random variables.\footnote{
The results of this paper is easily extendable to the  case
where  the channel gains and noises are complex Gaussian random variables and the input is real.
This is due to the fact that,  each use of the complex channel can
be thought of as two independent uses of a real additive white Gaussian noise channel,
noting that the noise is independent in the I and Q components  \cite{tse2005fundamentals}.
}
The received signal at the legitimate receiver and eavesdropper are, respectively, given by
\begin{subequations}\label{eq:sys}
\begin{align}
\mathbf{y}_r = \mathbf{H} \, \mathbf{x} \, + \, \mathbf{w}_r,  \\
\mathbf{y}_e = \mathbf{G} \, \mathbf{x} \, + \,  \mathbf{w}_e,
\end{align}
\end{subequations}
in which $\mathbf{x} \in \mathbb{R}^{n_t\times 1}$ is the transmitted signal and  $\mathbf{w}_i \in \mathbb{R}^{n_i\times 1}$, $i\in\{r,e\}$,
represents an i.i.d. Gaussian noise vector with zero mean and identity covariance matrix.
As will be seen later,  $\mathbf{x}=\mathbf{V}  \mathbf{s}$ where  $\mathbf{V}\in \mathbb{R}^{n_t \times n_t}$
 is the \textit{precoding matrix} to transmit a secrete data symbol vector  $\mathbf{s}$.
The transmitted signal is subject to an average power constraint
\begin{align*}
\mathrm{tr} ( \mathbb{E} \{ \mathbf {xx}^t \} ) = \mathrm{tr}(\mathbf {Q}) \le P,
\end{align*}
where $P$ is a scalar, and $\mathbf {Q} =  \mathbb{E} \{ \mathbf {xx}^t \}$ is the input covariance matrix.

A single-letter expression for the secrecy capacity of
the  general \textit{discrete memoryless} wiretap channel with transition probability $p(y_r,y_e|x)$
is given by \cite{csiszar1978broadcast}
\begin{align}\label{eq:capDMC}
C_{s} = \max _{p(u,x)}  \big[I(U;Y_{r})-I(U;Y_{e})\big],
\end{align}
in which the auxiliary random variable $U$ satisfies the Markov relation
 $U \rightarrow X \rightarrow (Y_{r},Y_{e}).$

With this, the problem of characterizing the secrecy capacity of the multiple-antenna wiretap channel
reduces to evaluating \eqref{eq:capDMC} for the channel model given in \eqref{eq:sys}. This
  was, however, open until the  work of Khisti and Wornell \cite{khisti2010secure} and Oggier and Hassibi \cite{oggier2011secrecy},
  where they proved that $U=X$ is optimal in \eqref{eq:capDMC}. Then, the secrecy capacity (bits per real dimension)
 is the solution of the following optimization problem~\footnote{For a complex channel, the factor $\frac{1}{2}$ is dropped as the capacity
 per complex dimension  is twice as the capacity per real dimension}\cite{khisti2010secure,oggier2011secrecy,liu2009note}:

\begin{equation}
\begin{aligned}\label{eq:cap0}
&\max_{\mathbf{Q}}
& &\frac{1}{2} \left[  \log \det(\mathbf{I}_{n_r}\!\! + \mathbf{H} \mathbf{Q} \mathbf{H}^t) -  \log \det(\mathbf{I}_{n_e} \!\! + \mathbf{G} \mathbf{Q} \mathbf{G}^t) \right] \\
 &\operatorname{s.t.} &&\mathbf{Q}  \succeq \mathbf{0}, \mathbf{Q}  = \mathbf{Q}^t, \; \mathrm{tr} (\mathbf{Q}) \le P ,
\end{aligned}
\end{equation}
in which the first two constraints are due to the fact that $\mathbf{Q}$ is a covariance matrix
  and the third constraint is the aforementioned average power constraint. 
 The secrecy capacity is obviously  nonnegative  as $\mathbf{Q}= \mathbf{0}$ is a feasible solution of \eqref{eq:cap0}.
The above optimization problem is non-convex (except for   $n_r = n_e = 1$ \cite{li2007secret})
 and its objective function possesses numerous local maxima \cite{bashar2012secrecy,li2013transmit,loyka2015algorithm}.
 As such, a closed-form solution for the optimum  $\mathbf{Q}$ is  not known, in general.

The problem of characterizing the optimal input covariance matrix that achieves secrecy capacity subject
to a power constraint has been under active investigation recently \cite{loyka2012optimal,fakoorian2013full,loyka2016optimal,li2009transmitter}.
Until recently, the special cases for which the optimal  $\mathbf{Q}$ was
known were limited to the cases of
$n_r =1$\cite{khisti2010secureMISOME} and
$n_t=2$, $n_r=2$, $n_e = 1$  \cite{shafiee2009towards}.\footnote{ In  these cases, the capacity is obtained  by  beamforming
(i.e., signaling with rank one covariance) along  the direction
of the  generalized eigenvector of $\mathbf{H}$ and $\mathbf{G}$ corresponding to the maximum eigenvalue of that pair.}
More recently, major steps have been made in characterizing the optimal covariance matrix. 
Fakoorian and Swindlehurst  \cite{fakoorian2013full} determined conditions under which the optimal
input covariance matrix is full-rank or rank-deficient. They also
fully characterized the optimal $\mathbf{Q}$ when it is full-rank.
Very recently, Loyka and   Charalambous \cite{loyka2016optimal} found a closed-form solution for
the optimal covariance matrix when the channel is strictly degraded ($\mathbf{H}^H\mathbf{H}\succ \mathbf{G}^H\mathbf{G}$)
 and  transmission power is greater than a certain value.
The combination of this result and the unit-rank solution of \cite{khisti2010secure} gives the optimal
 $\mathbf{Q}$ for the rank-2 case \cite{loyka2016optimal}. The optimal solution is, however,
still open in general.

In this paper, we study the MIMO wiretap channel with $n_t=2$  while $n_r$ and $n_e$ are arbitrary integers.
We derive a closed-form solution for the optimal covariance matrix in this case.
Our approach is completely different from  that of \cite{fakoorian2013full} and  \cite{loyka2016optimal}.
In addition, unlike \cite{fakoorian2013full} and  \cite{loyka2016optimal},
 our approach does not require finding the rank of the optimal covariance matrix
before fully characterizing the solution. It gives the optimal solution for both
full-rank and rank-deficient cases in one shot. What is more,
in \cite{loyka2016optimal}, it is
not clear when the rank of the optimal solution switches from one to two (i.e., the paper does not clarify at what power threshold this
change of rank happens); thus, it is not known whether a rank-one solution or full-rank solution should be applied.

\section{A Capacity Achieving Precoding}
\label{sec:main}

Based on the optimization problem in  \eqref{eq:cap0}, a characterization of the secrecy capacity of the MIMO Gaussian wiretap channel is given by
 non-negative $R$ such that
\begin{align} \label{eq:cap}
R &\le \max_{\mathbf{Q}}  \frac{1}{2} \left[  \log \det(\mathbf{I}_{n_r}\!\! + \mathbf{H} \mathbf{Q} \mathbf{H}^t) -  \log \det(\mathbf{I}_{n_e} \!\! + \mathbf{G} \mathbf{Q} \mathbf{G}^t) \right] \notag \\
&= \max_{ \mathbf{Q}} \frac{1}{2}  \log \frac{\det(\mathbf{I}_{n_t} \!\! + \mathbf{H}^t\mathbf{H} \mathbf{Q} ) }{ \det(\mathbf{I}_{n_t} \!\!+ \mathbf{G}^t \mathbf{G} \mathbf{Q} ) },
\end{align}
\noindent where $\mathbf{Q}  \succeq \mathbf{0}, \mathbf{Q}  = \mathbf{Q}^t, \mathrm{tr} (\mathbf{Q}) \le P$.
The equality in \eqref{eq:cap}  is due to the fact  that for any
$\mathbf{A} \in \mathbb{C}^{m\times n}$ and $\mathbf{B} \in \mathbb{C}^{n\times m}$ we have
\begin{align} \label{eq:AB}
\det(\mathbf{I}_{m} + \mathbf{A} \mathbf{B}) = \det(\mathbf{I}_{n} + \mathbf{B} \mathbf{A}).
\end{align}
Note that $\mathbf{H}^t\mathbf{H}$ and $\mathbf{G}^t \mathbf{G}$ are $n_t \times n_t $ symmetric matrices.
Also, $\mathbf{Q}$ is an $n_t \times n_t $ symmetric matrix and its \textit{eigendecomposition}
can be written as
\begin{align}  \label{eq:Q}
 \mathbf {Q} =\mathbf{V} \mathbf {\Lambda } \mathbf{V}^t,
\end{align}
where $\mathbf{V}\in \mathbb{R}^{n_t \times n_t}$ is the \textit{orthogonal matrix} whose $i$th column is
the $i$th \textit{eigenvector} of $\mathbf {Q}$ and $\mathbf {\Lambda}$ is the diagonal matrix whose
diagonal elements are the corresponding eigenvalues, i.e., $\mathbf {\Lambda}_{ii}=\lambda _{i}$.
In this paper, we study the case where $n_t=2$  while $n_r$ and $n_e$ are arbitrary integers.
\subsection{Reformulating the Problem for $n_t=2$}
\label{sec:refor}
We simplify the optimization problem   \eqref{eq:cap} for $n_t=2$  in this subsection.
Since $\mathbf{V}$ is  orthogonal  its columns are orthonormal and, without loss of generality,  we can write
\begin{align} \label{eq:V}
\mathbf{V} = \left[ \begin{matrix}
-\sin \theta & \cos \theta \\ \cos \theta & \sin \theta
\end{matrix} \right],
\end{align}
for some $\theta$.
Further, let
\begin{align}\label{eq:HG}
\mathbf{H}^t\mathbf{H} = \left[ \begin{matrix}
h_1 & h_2 \\ h_2 & h_3
\end{matrix} \right], \quad
\mathbf{G}^t\mathbf{G} = \left[ \begin{matrix}
g_1 & g_2 \\ g_2 & g_3
\end{matrix} \right].
\end{align}
The following lemma converts the optimization problem   \eqref{eq:cap} into a more tractable problem.
\begin{lem}\label{lem:equi-opt}
For $n_t = 2$ but arbitrary $n_r$ and $n_e$, the optimization problem in \eqref{eq:cap} is equivalent to
\begin{align} \label{eq:cap2}
R \le \max_{  \lambda_1 + \lambda_2 \le P } \frac{1}{2}   \log \left(\frac{ a_1 \sin2\theta + b_1 \cos2\theta +c_1 }{a_2 \sin2\theta + b_2 \cos2\theta +c_2}\right),
\end{align}
in which $\lambda_1$ and $\lambda_2$ are nonnegative, and
\begin{subequations}\label{eq:num}
\begin{align}
a_1 &=  (\lambda_2 - \lambda_1)h_2 ,\\
b_1&= \frac{1}{2}(\lambda_1 - \lambda_2)(h_3 - h_1),\\
c_1 & = 1 + \frac{1}{2}(\lambda_1 + \lambda_2)(h_1 + h_3)   +\lambda_1 \lambda_2 (h_1 h_3- h_2^2),
\end{align}
\end{subequations}
and
\begin{subequations}\label{eq:den}
\begin{align}
a_2 &=  (\lambda_2 - \lambda_1)g_2 ,\\
b_2&= \frac{1}{2}(\lambda_1 - \lambda_2)(g_3 - g_1),\\
c_2 & = 1 + \frac{1}{2}(\lambda_1 + \lambda_2)(g_1 + g_3)   +\lambda_1 \lambda_2 (g_1 g_3- g_2^2).
\end{align}
\end{subequations}
\end{lem}

\begin{proof}
To prove this lemma, we  simplify the determinants in  \eqref{eq:cap}.
First,  consider $\det(\mathbf{I}_{n_t} + \mathbf{H}^t\mathbf{H} \mathbf{Q})$.
Using $\mathbf {Q}$ given in \eqref{eq:Q} and applying  \eqref{eq:AB},  it is seen that
$\det(\mathbf{I}_{n_t} + \mathbf{H}^t\mathbf{H} \mathbf{Q}) = \det(\mathbf{I}_{n_t} + \mathbf{V}^t\mathbf{H}^t\mathbf{H} \mathbf{V}\mathbf{\Lambda} ).$
Further, it is straightforward to check that
 \begin{align}
\mathbf{V}^t\mathbf{H}^t\mathbf{H} \mathbf{V} = \left[ \begin{matrix}
w_1 & w_2 \\ w_2 & w_3
\end{matrix} \right],
\end{align}
in which
\begin{subequations}\label{eq:wH}
\begin{align}
w_1 &= h_1\sin^2\theta + h_3\cos^2\theta - 2h_2\sin\theta \cos\theta,  \\
w_2 &= h_2(\cos^2\theta - \sin^2\theta) + (h_3 - h_1)\sin\theta \cos\theta,  \\
w_3 &= h_1\cos^2\theta + h_3\sin^2\theta + 2h_2 \sin\theta \cos\theta.
\end{align}
\end{subequations}
Consequently,
\begin{align}\label{eq:s1}
\det(\mathbf{I}_{n_t} + \mathbf{H}^t\mathbf{H} \mathbf{Q}) &= \det(\mathbf{I}_{n_t} + \mathbf{V}^t\mathbf{H}^t\mathbf{H} \mathbf{V}\mathbf{\Lambda} )  \\
&= (1 + \lambda_1w_1)(1 + \lambda_2w_3)-  \lambda_1 \lambda_2w_2^2. \notag
\end{align}
Next, using the basic trigonometric identities
\begin{subequations}\label{eq:identity}
\begin{align}
\cos 2\theta &= 2\cos^2\theta - 1 = 1- 2\sin^2\theta, \\
\sin 2\theta &= 2\sin\theta \cos \theta ,
\end{align}
\end{subequations}
it is straightforward to show that
\begin{subequations}\label{eq:ws}
\begin{align}
w_1
&=  \frac{h_1 + h_3}{2} +\frac{h_3 - h_1}{2} \cos 2\theta  -h_2\sin2\theta, \label{eq:ws1}\\
w_2
&= h_2\cos2\theta + \frac{h_3 - h_1}{2}\sin2\theta \label{eq:ws2},\\
w_3
& =  \frac{h_1 + h_3}{2} - \frac{h_3 - h_1}{2} \cos 2\theta  + h_2\sin2\theta \label{eq:ws4}.
\end{align}
\end{subequations}
Substituting \eqref{eq:ws1}-\eqref{eq:ws4} in \eqref{eq:s1}, we obtain
\begin{align}
\det(\mathbf{I}_{n_t} + \mathbf{H}^t\mathbf{H} \mathbf{Q})
& = a_1 \sin2\theta + b_1 \cos2\theta + c_1,
\end{align}
in which $a_1$,  $b_1$, and $c_1$ are given in \eqref{eq:num}.
Following similar steps it is clear that
\begin{align}
\det(\mathbf{I}_{n_t} + \mathbf{G}^t \mathbf{G} \mathbf{Q}) = a_2 \sin2\theta + b_2 \cos2\theta + c_2,
\end{align}
where $a_2$,  $b_2$, and $c_2$ are given in \eqref{eq:den}. It should be mentioned that the constraint $\lambda_1 + \lambda_2 \le P$
comes from  $\mathrm{tr} (\mathbf{Q}) \le P$ since, from \eqref{eq:Q},
$\mathrm{tr} (\mathbf{Q}) = \mathrm{tr} (\mathbf{V} \mathbf {\Lambda } \mathbf{V}^t )
= \mathrm{tr} (\mathbf{V}^t \mathbf{V} \mathbf {\Lambda } ) = \mathrm{tr} (\mathbf {\Lambda } ).$
Note that $\mathrm{tr} ( \mathbf {AB}) = \mathrm{tr} ( \mathbf {BA})$  and  $\mathbf{V}^t \mathbf{V} = \mathbf {I}_{n_t}$. Also, $\lambda_1\ge 0$ and $\lambda_2\ge 0$ are due to $\mathbf{Q}  \succeq \mathbf{0}$.
This completes the proof of Lemma~\ref{lem:equi-opt}.
\end{proof}

\begin{figure*}[t]
\centering
\includegraphics[scale=.4]{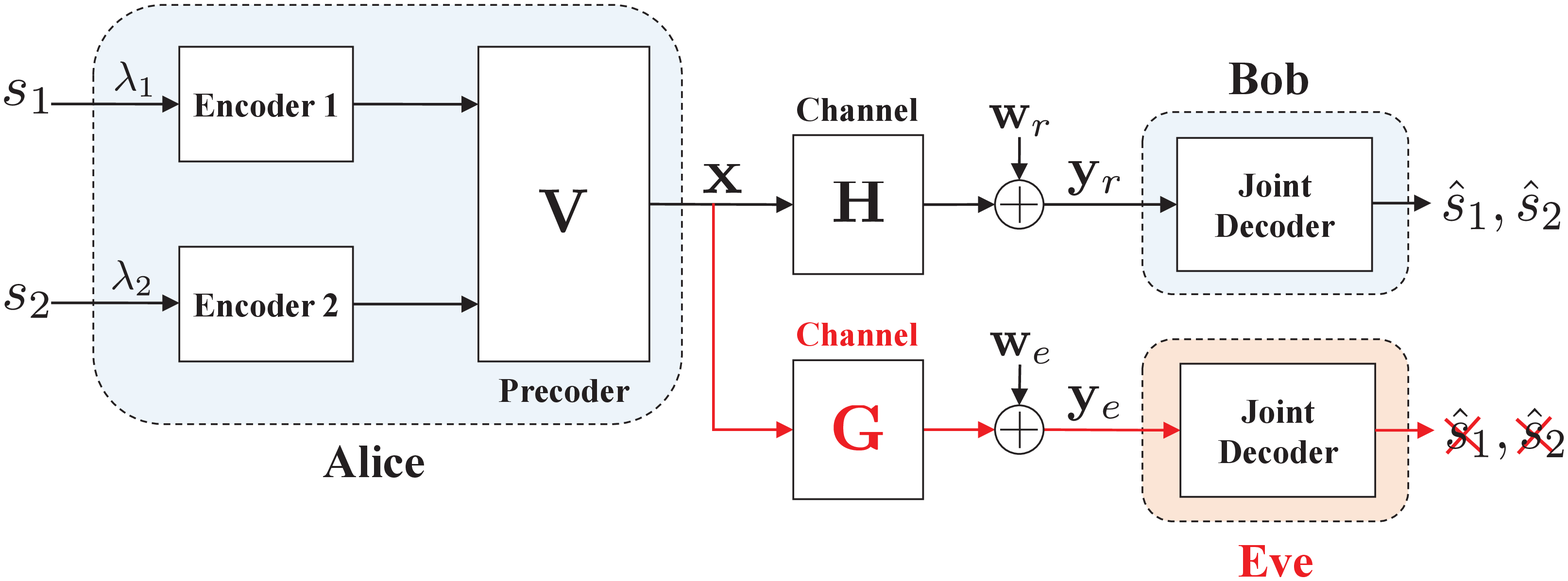}
\caption{Optimal architecture for communicating
over the MIMO Gaussian wiretap channel with $n_t= 2$ and arbitrary $n_r $ and $n_e$.
}
\label{fig:diagram}
\end{figure*}

\begin{lem}\label{lem:lsump}
In the optimization problem given by Lemma~\ref{lem:equi-opt}, the constraint $\lambda_1 + \lambda_2 \le P$
can be replaced either by $\lambda_1 + \lambda_2 = P$ or $\lambda_1 + \lambda_2 = 0$; i.e., it is optimal to use either all available power or nothing.
\end{lem}
\begin{proof}
See Appendix A.
\end{proof}

\subsection{Optimal Precoding}
\label{sec:pre}
In what follows, we first find a closed-form solution for the optimization problem
in Lemma~\ref{lem:equi-opt} for a given pair of $\lambda_1$ and $\lambda_2$
that satisfy the constraints.
Since $\log(x)$ is strictly increasing in $x$, we can instead maximize the argument of the logarithm in  \eqref{eq:cap2}.
Thus, let us define
\begin{align} \label{eq:W}
W = \frac{a_1\sin2\theta + b_1 \cos2\theta + c_1 }{a_2\sin2\theta + b_2 \cos2\theta + c_2 }.
\end{align}
Then, $\theta^* = \arg \max W$ and is obtained by differentiating
$W$ with respect to $\theta$ and finding its critical points. It can be checked that
 $\frac{\partial W}{\partial \theta}= 0$ is equivalent to
\begin{align}\label{eq:equation}
 a\sin2\theta + b\cos2\theta + c = 0,
\end{align}
in which
\begin{subequations}\label{eq:abc}
\begin{align}
a &= c_1 b_2 -c_2 b_1,  \\
b &= a_1 c_2 - a_2 c_1, \\
c &=  a_1 b_2 -a_2 b_1.
\end{align}
\end{subequations}
Before proceeding, we  note that $W$ is periodic in $\theta$ and its period is $\pi$.
Also, it can be checked that if both $a$ and $b$
are zero, then $\frac{a_1}{a_2} = \frac{b_1}{b_2} = \frac{c_1}{c_2} $ and $W$ is constant; i.e., any $\theta$
is optimal.  Thus, we assume  $a^2 + b^2 \neq 0.$
%
%
Defining  $\frac{b}{a} = \tan\phi$,
  \eqref{eq:equation} can be further simplified as
\begin{align}\label{eq:rondsimp}
\sin(2\theta + \phi) + \frac{c}{\sqrt{a^2+b^2}} = 0.
\end{align}
The critical points of the above equation are given by
\begin{align}\label{eq:thetaOPT1}
2\theta =
\begin{cases}
    - \arctan \frac{b}{a} - \arcsin \frac{c}{\sqrt{a^2+b^2}}+2k\pi  \\
   -  \arctan \frac{b}{a} + \pi + \arcsin \frac{c}{\sqrt{a^2+b^2}} +2k\pi  \\
\end{cases},
\end{align}
where $k$ is an integer.\footnote{It should be highlighted that we always have $|c|\le \sqrt{a^2+b^2}$, as
otherwise $W$ would be strictly increasing or strictly decreasing in $\theta$, which is impossible  because $W$ is periodic and continuous.
}
Then, using the second derivative of $W$ with respect to $\theta$,
 we can  verify that the first argument gives  the minimum  of $W$ while the second one gives  its maximum.
For completeness, this is proved in Appendix~B.  Further, without loss of optimality, we let $k=0$ in \eqref{eq:thetaOPT1}.
Hence, the optimal $\theta$ that  maximizes $W$ is obtained by
\begin{align}\label{eq:thetaOPT}
\theta^* =   -\frac{1}{2}  \arctan \frac{b}{a}   +\frac{1}{2} \arcsin \frac{c}{\sqrt{a^2+b^2}} +\frac{\pi}{2}.
\end{align}

Thus far, the optimal $\theta$ is obtained for given $\lambda_1$ and $\lambda_2$. To find the
optimal  $\lambda_1$  and $\lambda_2$, in light of  Lemma~\ref{lem:lsump}, we can
search over all $\lambda_1\ge0$ and $\lambda_2\ge0$  that satisfy
$\lambda_1 + \lambda_2 =P $ or $\lambda_1 + \lambda_2 =0 $ and  maximize  \eqref{eq:W} where $\theta$ is given in \eqref{eq:thetaOPT}.
We can vary $\lambda_1$ from $0$ to $P$. Therefore, we have the following.

\begin{thm}\label{thm:mimome}
 To achieve the secrecy capacity of the MIMO Gaussian
wiretap channel (with $n_t = 2$) under the average power constraint $P$, it  suffices to use
\begin{align}
\mathbf{V} = \left[ \begin{matrix}
-\sin \theta & \cos \theta \\ \cos \theta & \sin \theta
\end{matrix} \right],
\end{align}
as the transmit beamformer 
with the power allocation matrix
\begin{align}
\mathbf{\Lambda} = \left[ \begin{matrix}
\lambda_1 & 0 \\ 0 & \lambda_2
\end{matrix} \right].
\end{align}
An optimal $\theta$ is given by \eqref{eq:thetaOPT} and is obtained by searching over nonnegative $\lambda_1$ and $\lambda_2$
that satisfy $\lambda_1+ \lambda_2 = P$ or $\lambda_1+ \lambda_2 = 0$ and maximize \eqref{eq:W}.
\end{thm}

Once the optimal $\mathbf{V}$,  $\lambda_1$, and $\lambda_2$ are determined, these can be used
for precoding and power allocation as illustrated in Fig~\ref{fig:diagram},
similarly to the  V-BLAST architecture for communicating
over the MIMO channel \cite{tse2005fundamentals}.
Here, two ($n_t = 2$) independent data streams  are multiplexed in the
coordinate system given by the precoding matrix $\mathbf{V}$.  The $i$th data stream is allocated a power $\lambda_i$.
Each stream is encoded using a capacity-achieving Gaussian code.
The data streams are decoded jointly.
When the orthogonal matrix $\mathbf{V}$ and powers $\lambda_i$  are chosen as described in Theorem~\ref{thm:mimome}, then we
have the capacity-achieving architecture in Fig~\ref{fig:diagram}.\footnote{It is worth mentioning that we can come up with another orthogonal matrix $\mathbf{U}$, to
express the output in terms of its columns,  such that the input/output relationship
is very simple and independent decoding is optimal. }

\begin{lem}\label{lem:sym}
With a proper choice of $\theta$, the pairs $(\lambda_1 ,\lambda_2)$ and $(\lambda_2 ,\lambda_1)$
result in the same maximum rate in Lemma~\ref{lem:equi-opt}.
\end{lem}
\begin{proof}
See Appendix C.
\end{proof}
This lemma implies that to find optimal $(\lambda_1 ,\lambda_2)$
in Theorem~\ref{thm:mimome}, it suffices to search for
$\lambda_1$  in $[0\; \frac{P}{2}]$ rather than $[0\; P]$.
\subsection{Special Cases}
\label{sec:spe}
The first special case of the MIMO  Gaussian wiretap channel  we  consider  is  the MISO Gaussian wiretap channel.
In the following  corollary, we prove that a positive capacity for the MISO case is obtained by signaling with rank one covariance.
This has already been shown in \cite{khisti2010secureMISOME} using a different argument.  

\begin{cor}\label{cor:misome}
For the MISO Gaussian wiretap channel, Theorem~\ref{thm:mimome} significantly simplifies and
 ($\lambda_1, \lambda_2) = (0, P)$ or ($\lambda_1, \lambda_2) = (0, 0)$ is the optimal solution.
The optimal $\theta$ is then obtained from \eqref{eq:thetaOPT}. 
\end{cor}

 \begin{proof}
In the case of the MISO multi-eavesdropper  wiretap channel it is known that the rank of the covariance matrix is either one or zero
(see  \cite[Theorem~2]{khisti2010secureMISOME} or \cite{loyka2016optimal}).
 In the latter case, it is trivial that ($\lambda_1, \lambda_2) = (0, 0)$ is an optimal solution. In the former case,
from Theorem~\ref{thm:mimome} we can see that a rank-one solution implies that  either $\lambda_1$ or $\lambda_2$ is equal to zero.
Then, from Lemma~\ref{lem:lsump} we conclude that ($\lambda_1, \lambda_2) = (P, 0)$ or ($\lambda_1, \lambda_2) = (0, P)$.
But, in view of Lemma~\ref{lem:sym}, we know that with proper choice of $\theta$ these two cases result in the same maximum rates; thus, one of them can be removed.
\end{proof}

Another special case of the MIMO Gaussian wiretap channel
is the case in which  the eavesdropper has only one antenna. Specifically, by setting
$n_e = 1$ in Theorem~\ref{thm:mimome} we  get
\begin{cor}\label{cor:mimose}
For the $2$-$n_r$-$1$ Gaussian wiretap channel, optimal transmit covariance matrix is at most unit-rank.
In particular,  either ($\lambda_1, \lambda_2) = (0, P)$ or ($\lambda_1, \lambda_2) = (0, 0)$ gives the optimal solution in Theorem~\ref{thm:mimome}.
\end{cor}
\begin{proof}
The proof is very similar to that of Corollary~\ref{cor:misome} and is omitted.
Note that the objective function, in this case, is in the form of the inverse of that of Corollary~\ref{cor:misome}.
\end{proof}
Note that Corollary~\ref{cor:misome} gives the capacity of $2$-$n_r$-$1$ channels and thus generalizes
 the  result of \cite{shafiee2009towards} for the $2$-$2$-$1$ channel.

\subsection{Closed-Form Solution for Optimal Power Allocation}
\label{sec:analy}

Finding optimal $\lambda_1$ and $\lambda_2$ in Theorem~\ref{thm:mimome} requires an exhaustive search.
Although checking a reasonably small number of ($\lambda_1$, $\lambda_2$)
is enough in practice,\footnote{This is discussed in Section~\ref{sec:sim}.}
in this subsection  we find  a closed-form solution for optimal ($\lambda_1$, $\lambda_2$).

We know that if  $W \leq 1$ then $(\lambda^*_1 ,\lambda^*_2) = (0, 0) $ is the optimal solution.
Thus, let us assume $W > 1$. Then, using Lemma~\ref{lem:lsump}, this  implies that  $\lambda_1 + \lambda_2 = P$ is optimal.
Thus, to find optimal $\lambda_1$ and $\lambda_2$, we can solve the following problem:
\begin{align} \label{eq:capMIMOME}
\mathcal{C}_\mathrm{MIMOME} =  \max_{  \lambda_1 + \lambda_2 = P } \frac{1}{2}   \log (W),
\end{align}
where $W = \det(\mathbf{I}_{n_t} + \mathbf{H}^t\mathbf{H} \mathbf{Q})/\det(\mathbf{I}_{n_t} \!\!+ \mathbf{G}^t \mathbf{G} \mathbf{Q}) $ is given in \eqref{eq:cap}.
To this end, we  define $ a_h \triangleq \frac{h_3 - h_1}{2},  b_h \triangleq -h_2,$
$c_h \triangleq \frac{h_1 + h_3}{2}, d_h \triangleq \sqrt{a_h^2 + b_h^2},$ and  $\frac{b_h}{a_h} \triangleq \tan\phi_h$.
Then, from  \eqref{eq:ws1}-\eqref{eq:ws4}  we will have
\begin{subequations}\label{eq:wssimple}
\begin{align}
w_1 & = c_h + d_h \cos (2\theta -\phi_h), \label{eq:ws1simple} \\
w_2 & = d_h \sin (2\theta -\phi_h), \label{eq:ws2simple} \\
w_3 & = c_h - d_h \cos (2\theta -\phi_h)\label{eq:ws3simple}.
\end{align}
\end{subequations}

\noindent Now, we can write
\begin{align}\label{eq:Wh}
W_h  &= \det(\mathbf{I}_{n_t}\!\! + \mathbf{H}^t\mathbf{H} \mathbf{Q}) \notag \\
&  \stackrel{(a)}{=}(1 + \lambda_1w_1)(1 + \lambda_2w_3)-  \lambda_1 \lambda_2w_2^2 \notag \\
&= 1 + \lambda_1 w_1  + \lambda_2 w_3 +  \lambda_1 \lambda_2 (w_1 w_3 -w_2^2) \notag \\
&  \stackrel{(b)}{=} 1 + \lambda_1 w_1  + \lambda_2 w_3 +  \lambda_1 \lambda_2 (h_1 h_3- h_2^2) \notag \\
& \stackrel{(c)}{=} 1 + (\lambda_1\! + \!\lambda_2)c_h +(\lambda_1\! - \!\lambda_2)d_h\cos (2\theta -\phi_h)  \nonumber \\  &\quad +  \lambda_1 \lambda_2 (h_1 h_3- h_2^2), \notag \\
& \stackrel{(d)}{=} 1 + Pc_h +(2 \lambda_1 - P)d_h\cos (2\theta -\phi_h)  \nonumber \\  &\quad +  \lambda_1 (P -\lambda_1) (h_1 h_3- h_2^2) \notag  \\
& \stackrel{(e)}{=} \alpha_h + \beta_h \lambda_1 -  \delta_h \lambda_1^2,
\end{align}
in which  $(a)$ is due to \eqref {eq:s1}, $(b)$  can be verified using \eqref{eq:ws1}-\eqref{eq:ws4},
 $(c)$ is due to \eqref{eq:ws1simple} and \eqref{eq:ws3simple},  $(d)$ is due to the fact that
  $\lambda_1 + \lambda_2 = P$ is optimal when $W > 1$, which follows from Lemma~\ref{lem:lsump}, and
  $(e)$ is obtained by defining
  \begin{subequations}
\begin{align} \label{eq:abh}
\alpha_h &= 1 + Pc_h - Pd_h\cos (2\theta -\phi_h),  \\
\beta_h & = 2 d_h\cos (2\theta -\phi_h) + P\delta_h, \\
\delta_h & = h_1 h_3- h_2^2.
\end{align}
\end{subequations}
In a  similar way, we can show that
\begin{align} \label{eq:Wg}
W_g &= \det(\mathbf{I}_{n_t} \!\!+ \mathbf{G}^t \mathbf{G} \mathbf{Q})= \alpha_g + \beta_g \lambda_1 -  \delta_g \lambda_1^2,
\end{align}
where
\begin{subequations}
\begin{align} \label{eq:abg}
\alpha_g &= 1 + Pc_g - Pd_g\cos (2\theta -\phi_g),  \\
\beta_g & = 2 d_g\cos (2\theta -\phi_g) + P\delta_g, \\
\delta_g & = g_1 g_3- g_2^2,
\end{align}
\end{subequations}
and $c_g, d_g$, and $ \phi_g$ are defined for $\mathbf{G}$ similarly to those of $\mathbf{H}$.
   Hence, we can write
\begin{align} \label{eq:WMIMOME2}
W &= \frac{W_h}{W_g} = \frac{ \alpha_h + \beta_h \lambda_1 -  \delta_h \lambda_1^2}{  \alpha_g + \beta_g \lambda_1 -  \delta_g \lambda_1^2}.
\end{align}
Next,  it can be checked that
\begin{align}\label{eq:parWMIMOME2}
\frac{\partial W}{\partial \lambda_1} = \frac{\bar c + \bar b \lambda_1  + \bar a \lambda_1^2}{(  \alpha_g + \beta_g \lambda_1 -  \delta_g \lambda_1^2)^2},
\end{align} in which
\begin{subequations}
\begin{align}
\bar a &= \delta_g  \beta_h  - \delta_h  \beta_g, \\
\bar b &=  2\delta_g \alpha_h - 2\delta_h \alpha_g, \\
\bar c &= \beta_h\alpha_g - \beta_g\alpha_h.
\end{align}
\end{subequations}
Let  $\Delta = \bar b^2 -4 \bar a \bar c$,  and suppose that $\Delta > 0$.\footnote{ When $\Delta \le 0$, $W$  is  strictly decreasing
or increasing    with $\lambda_1$, and    $\lambda_1=0$ or $\lambda_1 =P$  are the only critical points.} Then
\begin{subequations}\label{eq:parroots}
\begin{align}
\lambda_{1,1}^* = (- \bar b + \sqrt{\Delta})/2 \bar a, \label{eq:parroot1}\\ 
\lambda_{1,2}^* = (- \bar b - \sqrt{\Delta})/2\bar a,  \label{eq:parroot2}
\end{align}
\end{subequations}
are the roots of \eqref{eq:parWMIMOME2}. Next, it is easy to show that, for
$\lambda_{1,i}^*, \; i\in\{1,2\}$, in  \eqref{eq:parroot1} and \eqref{eq:parroot2} we have
\begin{align}\label{eq:par2WMIMOME2}
\frac{\partial^2 W}{\partial \lambda_1^2}(\lambda_{1,i}^*) &= \frac{\bar b  + 2 \bar a \lambda_{1,i}^*}{( \alpha_g + \beta_g \lambda_1 -  \delta_g \lambda_1^2)^2}
 = \begin{cases}
   + \frac{\sqrt{\Delta}}{W_g^2}, \qquad   i=1\\
  - \frac{\sqrt{\Delta}}{W_g^2}, \qquad   i=2\\
\end{cases}\!\!\!\!\!\!.
\end{align}
That is, the second derivative is positive at $\lambda_{1,1}^*$  and negative at $\lambda_{1,2}^*$. Thus, the former
corresponds to a minimum  of $W$ and the latter corresponds to a maximum  of that quantity.
Therefore, the following cases appear:
 \subsubsection{Case I ($\Delta \le 0)$} This case  results in a strictly decreasing  or increasing $W$  in $\lambda_1$. Then,
   $\lambda_1=0$ or $\lambda_1 =P$  is optimal, depending on the sign of $a$.
   The optimum value of $\lambda_1$ can be inserted
  into \eqref{eq:num} and \eqref{eq:den} to find the  optimal $\theta$.
  The optimal $\lambda_2$ is obtained from $\lambda_1 + \lambda_2 =P$.

%

\subsubsection{Case II ($\Delta > 0)$}
In this case, the maximum of $W$ is achieved by
  $\lambda_1 = 0$,  $\lambda_1= P$, or $\lambda_1= \lambda_{1,2}^*$, provided that $ 0\le \lambda_{1,2}^* \le P$.
  The optimal $\lambda_2$ is obtained from $\lambda_1 + \lambda_2 =P$.
Hence, when $W >1$, $(\lambda^*_1 ,\lambda^*_2)$ is one of the following pairs:
$(0, P)$,  $(P, 0)$, or $(\lambda_{1,2}^*,  P-\lambda_{1,2}^*)$. But,
in light of Lemma~\ref{lem:sym}, it can be seen that
 $(0, P)$ and   $(P, 0)$ result in the same optimum $W $ and thus one of them can be omitted.

To summarize, considering all cases for  $W \leq 1$ and  $W >1$,
it is enough to check
\begin{subequations}\label{eq:lambda1s}
\begin{align}
(\lambda^*_1 ,\lambda^*_2) &= (0, 0),  \label{eq:lambda1s1} \\
(\lambda^*_1 ,\lambda^*_2) &= (0, P),  \label{eq:lambda1s2} \\
(\lambda^*_1 ,\lambda^*_2) &= (\lambda_{1,2}^*,  P-\lambda_{1,2}^*), \label{eq:lambda1s4}
\end{align}
\end{subequations}
in order to obtain the maximum of $W$.
 We should highlight that \eqref{eq:lambda1s4} will be a choice only if $\lambda_{1,2}^*$,
defined in  \eqref{eq:parroot2}, is a real number between $ 0 $ and $ P$.
As a result, we have
\begin{thm}\label{thm:mimome2}
The optimal $\lambda_1$ and $\lambda_2$  in Theorem~\ref{thm:mimome} is confined to one of the following cases:
\begin{align}\label{eq:lOPTMIMOME}
(\lambda_1, \lambda_2)=
\begin{cases}
   ( 0, 0), \qquad \;  \\ 
   ( 0, P),  \qquad \\ 
   (\lambda^*,  P-\lambda^*) , \qquad \\ 
\end{cases},
\end{align}
in which $\lambda^* \triangleq \lambda_{1,2}^*$ is defined in  \eqref{eq:parroot2}, and  $\theta$ is given in \eqref{eq:thetaOPT}.
\end{thm}
\begin{rem}
As can be  traced  from   \eqref{eq:parroot2},  in general, the optimal $\lambda_1$ is a function of $\theta$.
On the other hand, the optimal $\theta$, given in \eqref{eq:thetaOPT},  is a function of $\lambda_1$ (and $\lambda_2$).
Thus, the triplet $(\lambda_1, \lambda_2, \theta)$ can be found   for any possible maximizing argument in \eqref{eq:lOPTMIMOME}.
Then, by evaluating $W$ for these points we can determine which one is the optimal (capacity-achieving) solution.
For the first two cases in \eqref{eq:lOPTMIMOME} the solution is obtained analytically. However, the equation resulting from
combining  the third case in \eqref{eq:lOPTMIMOME} and \eqref{eq:thetaOPT} is rather cumbersome and thus we solve it numerically.
\end{rem}

\section{Special Cases and Possible Extensions}
\label{sec:ext}
In this section, we briefly consider some special cases of the proposed precoding as well as possible
extensions of this work.

\subsection{Beamforming for MISO and MIMO Channels}
\label{sec:mimo}
The optimal beamforming  provided in the previous section
achieves the capacity of  MISO and MIMO channels without an eavesdropper ($\mathbf{G} =  \mathbf{0}$),
as shown below.

\subsubsection{Capacity of MISO Channels}
We know that the capacity of a MISO channel is given by \cite{tse2005fundamentals}
\begin{align} \label{eq:capMISO}
\mathcal{C}_\mathrm{MISO} =   \frac{1}{2}   \log (1 + \Vert \mathbf{h} \Vert ^2 P ),
\end{align}
where $\mathbf{h}$ is the channel vector.
On the other hand, using  \eqref{eq:s1}, it is straightforward to check that the above rate
  is achieved by letting $\lambda_1 = P$, $\lambda_2 = 0$,  and $\theta = \frac{\pi}{2} + \alpha$, where $\tan \alpha\triangleq \frac{\sqrt{h_3}}{\sqrt{h_1}}$.

\subsubsection{Capacity of MIMO Channels}
\label{sec:MIMO}
It can be also checked that the proposed beamforming and power allocation is equal
to   SVD-based beamforming with water-filling for
$\theta = \frac{1}{2} \tan^{-1} \frac{b_h}{a_h} $
and
\begin{subequations}\label{eq:L1L2MIMO}
\begin{align}
\lambda_1 =  \min\Big\{\frac{P}{2}  + \frac{c_h}{\delta_h}, P \Big\}, \\
\lambda_2 =  \max\Big\{\frac{P}{2}  - \frac{c_h}{\delta_h}, 0 \Big\},
\end{align}
\end{subequations}
where $a_h = \frac{h_1 - h_3}{2}$, $b_h = h_2$,  $c_h = \sqrt{a_h^2+b_h^2}$, and $\delta_h = h_1h_1-h_2^2$.

\subsection{Extension to $n_t>2$}
The key idea in this paper is to use the fact that any  orthogonal matrix $\mathbf{V}$ is parametrized
by a single parameter $\theta$, as shown in  \eqref{eq:V}. Considering this, in \eqref{eq:cap}, we rewrite the capacity
expression in a way that  for any  $n_r$  and  $n_e$  (with $n_t=2$) the terms $\mathbf{H}^t\mathbf{H}$
and $\mathbf{G}^t\mathbf{G}$ are $2 \times 2$ matrices. Hence,  the capacity expression
can be represented by  three parameters, two nonnegative powers ($\lambda_1$ and $\lambda_2$) and
one angle $\theta$.\footnote{Excluding the case  $\mathbf{H}^H\mathbf{H}-\mathbf{G}^H\mathbf{G} \preceq 0$  which results in the trivial solution $(\lambda_1,\lambda_2)=(0,0)$,
from Lemma~\ref{lem:lsump} we can see that $\lambda_1 + \lambda_2 = P$. This implies that the capacity region
can be  expressed just by two parameters, i.e.,  $\lambda_1$ and $\theta$.} Then, the
covariance matrix can be  optimized with elementary trigonometric equations, as shown in Section~\ref{sec:main}. In the case
of $n_t = 3$, the main difficulty is to  parametrize the $3 \times 3$  orthogonal matrix $\mathbf{V}$  with two parameters.
Even with this, it is not guaranteed to get a tractable optimization problem.
We have made some progress  towards this goal, but the resulting
optimization problem is rather cumbersome and needs further simplification.
This issue becomes more challenging as  $n_t$ increases.

\subsection{Construction of Practical Codes}

Although the capacity of the MIMO wiretap channel  is well-studied,
construction of practical codes is still a challenging issue for this channel.
Recently, it has been shown in \cite{khina2014ordinary} that a good wiretap code,
e.g., a scalar random-binning code  \cite{tyagi2014explicit}, is applicable to
the MIMO wiretap channel in conjunction with a linear encoder and a
successive interference cancellation  (SIC) decoder to achieve a rate
close to the MIMO wiretap capacity. However, this
approach  gives rise to several practical issues in terms
of implementation, such as dithering  in the SIC decoder.
Considering this, one direction for future work would be to find a more practical code
construction for MIMO wiretap channels based on our new design
of closed-form optimal beamforming and power allocation solutions.

\begin{figure*}[t]
\centering
\subfigure[$n_e =1$]{
\includegraphics[scale=.59]{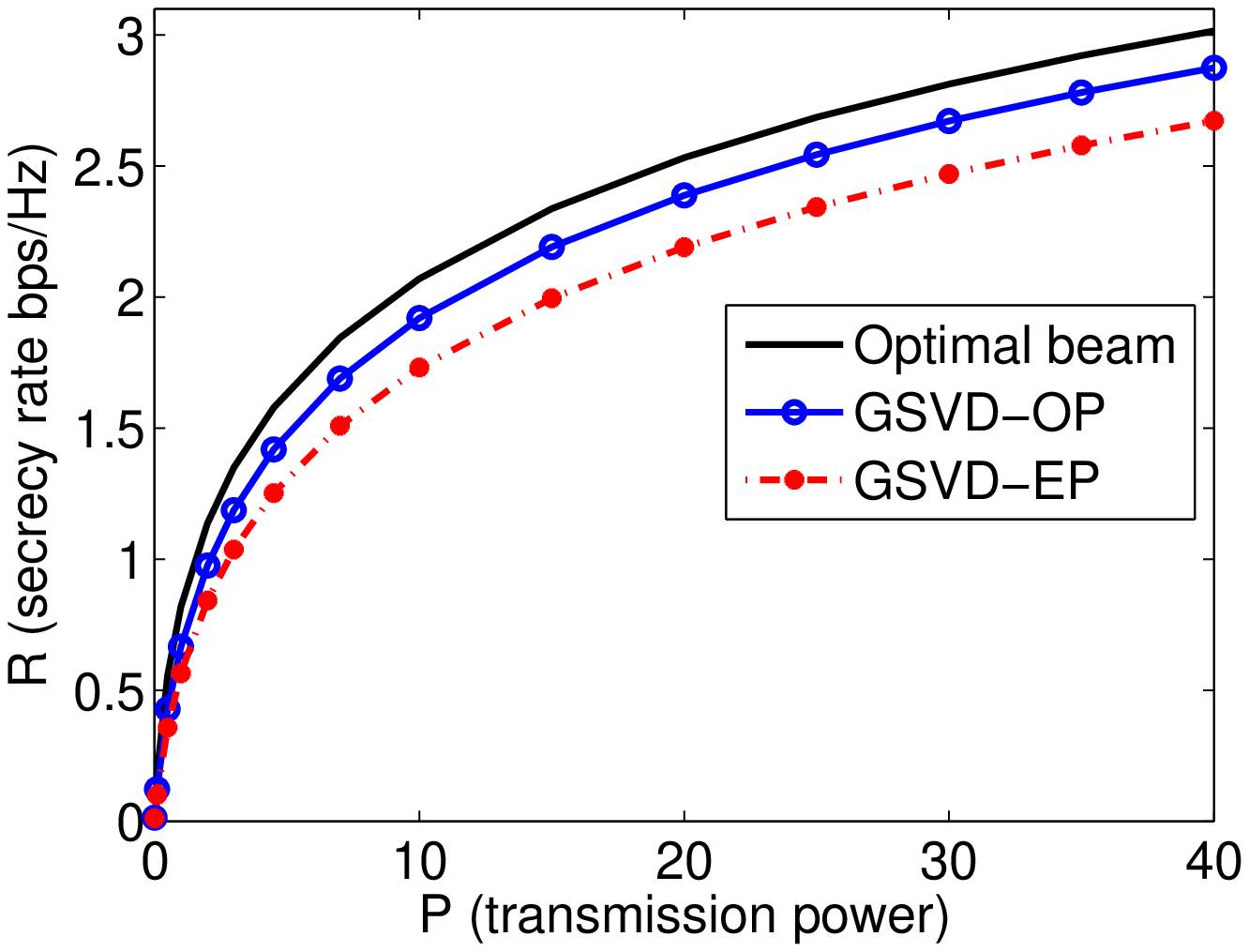}
\label{fig:MIMOME221}
}
\subfigure[$n_e =2$]{
\includegraphics[scale=.59]{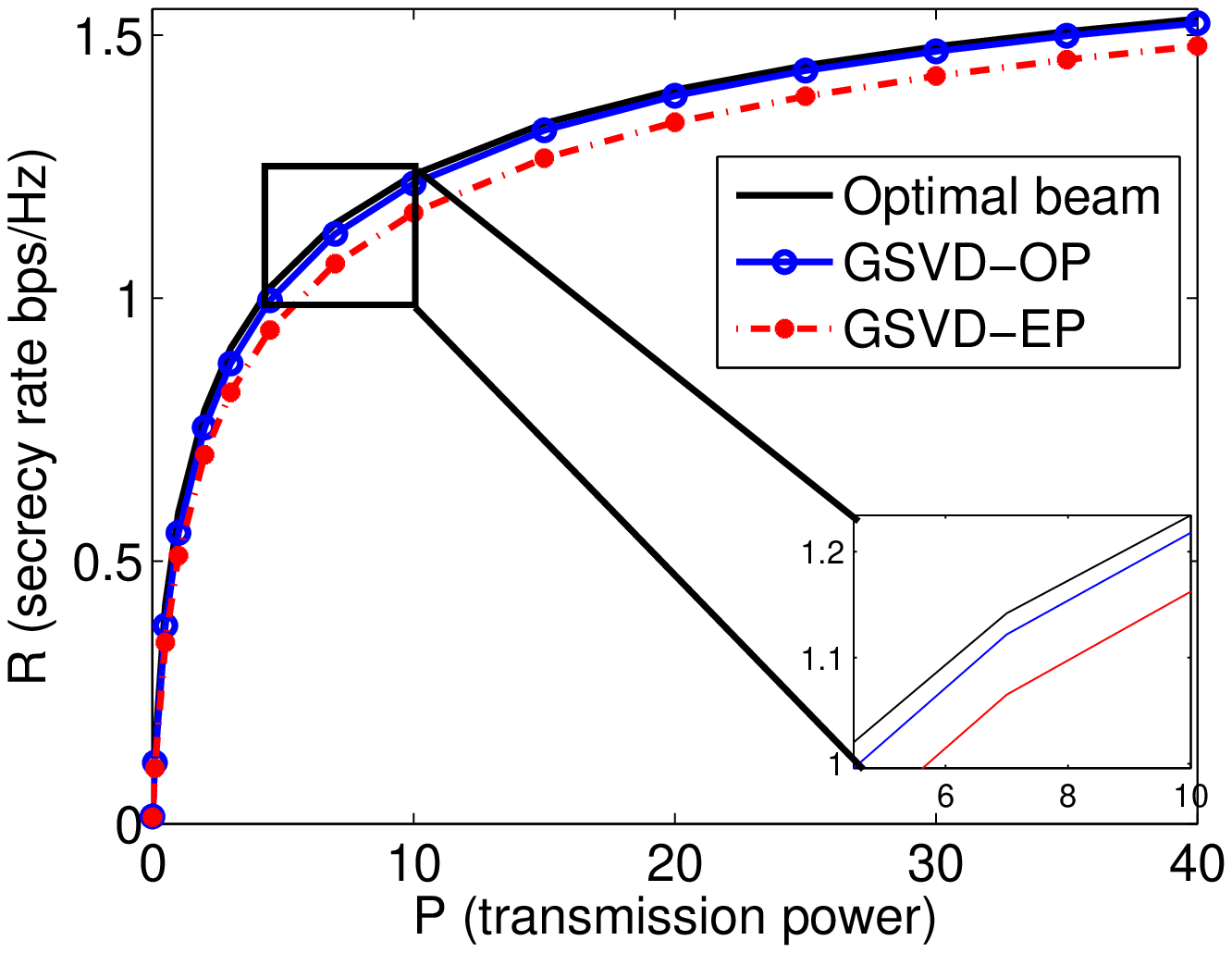}
\label{fig:MIMOME222}
}
\caption{Comparison of the secrecy capacity  of the MIMO Gaussian wiretap channel (achieved by the proposed beamforming method) and the  secrecy rate of GSVD-based beamforming
with equal and optimal power allocations for (a) $n_t=2$, $n_r =2$, $n_e =1$ (b) $n_t=2$, $n_r =2$, $n_e =2$.
}
\label{fig:MIMOME22}
\end{figure*}

\section{Numerical Results} \label{sec:sim}

In this section, we provide numerical examples to illustrate
 the  secrecy capacity of  Gaussian multi-antenna wiretap channels
using the proposed beamforming method.
We also compare our results with those of  GSVD-based beamforming with
equal power  (GSVD-EP) and optimal power (GSVD-OP) allocation proposed in
\cite{khisti2010secure} and \cite{fakoorian2012optimal}, respectively.
As  proved in Section~\ref{sec:main}, the proposed beamforming method is optimal and gives the capacity.
Numerical results are included here to show how much gain this optimal method brings  when compared with the
existing beamforming and power allocation methods.  It should be highlighted that the rate achieved by
GSVD-OP  is equal to or better than that of  GSVD-EP,
for any $\mathbf{H}$ and $\mathbf{G}$.

All simulation results are for  1000 independent realizations
of the channel matrices $\mathbf{H}$ and $\mathbf{G}$.
The entries of these matrices are generated  by i.i.d. $\mathcal{N}(0,1)$.
To get the capacity, we use the optimal power allocation of  Theorem~\ref{thm:mimome2}.
  We plot the secrecy rate versus total average power.

We first consider  the case with $n_t=2$, $n_r =2$, and $n_e =1$, where the eavesdropper has only one antenna.
As can be seen from Fig.~\ref{fig:MIMOME221}, the capacity-achieving beamforming performs significantly better than
both GSVD-based beamformings. By doubling  the eavesdropper's number of antennas in Fig.~\ref{fig:MIMOME222}, the secrecy capacity nearly  halves.
Moreover,
the rate achieved by the GSVD-OP becomes very close to that of optimal method.
However, as can be seen in Fig.~\ref{fig:MIMOME222}, there is still a small gap between the two methods particularly when $P$ is small.

We next consider the MISO wiretap channel in Fig.~\ref{fig:MIMOME212}. It can
be seen that there is a visible gap between the proposed beamforming and GSVD-based beamforming.
Note that GSVD-EP and GSVD-OP have exactly the same performance for  MISO wiretap channels.
This is because there is only one beam and all power is allocated to  that.
A general trend was seen  both for MISO and MIMO wiretap channels is that as SNR increases the
performance of GSVD-EP, and thus GSVD-OP,
get closer to that of  the optimal beamforming scheme derived in this paper. This is not surprising knowing that GSVD-EP
is asymptotically optimal; i.e., it is capacity-achieving as $P \to \infty $ \cite{khisti2010secure}.

%

\begin{figure}[t]
\centering
\includegraphics[scale=.59]{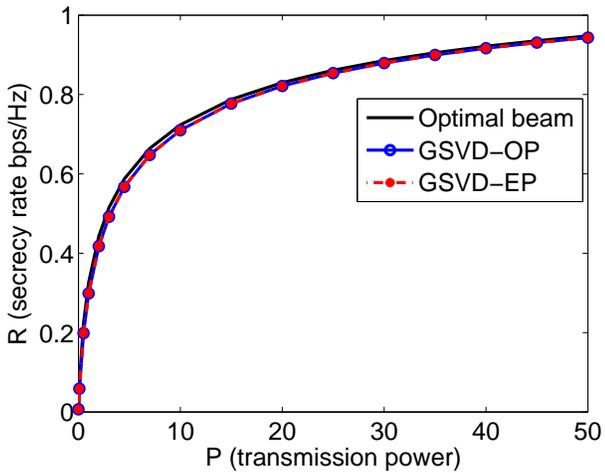}
\caption{The secrecy capacity  of the MISO wiretap channel  and the  secrecy rate of GSVD-based beamforming
 for $n_t=2$, $n_r =1$, and $n_e =2$.
}
\label{fig:MIMOME212}
\end{figure}

\begin{figure}[t]
\centering
\includegraphics[scale=.59]{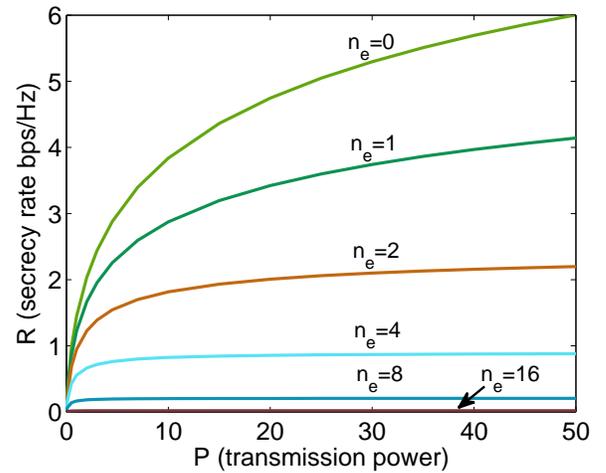}
\caption{The secrecy capacity  of the MIMO Gaussian wiretap channel for  various $n_e$, with $n_t=2$,  and $n_r=4$.
}
\label{fig:MIMOMEne}
\end{figure}

Figure~\ref{fig:MIMOMEne}  demonstrates the effect of increasing the number of antennas
at the eavesdropper. All curves in this figure are for $n_t=2$, $n_r=4$  but a different number of eavesdropper antennas, as depicted on each curve.
Note that $n_e = 0$ refers to the case where there is no eavesdropper; this curve is basically the capacity of MIMO
channel.\footnote{Recall from Section~\ref{sec:MIMO} that the proposed precoding for the MIMO Gaussian wiretap channel reduces to the
well-known SVD precoding of the MIMO channel  ($\mathbf{G} = \mathbf{0}$), and is capacity-achieving.}
Once the eavesdropper comes in play ($n_e \ge 1$),  the extent to which information can be secured over the air reduces.
The gap between each curve and the curve corresponding to  $n_e = 0$ is the unsecured information.
Unfortunately, for $n_e = 16$, and thus $n_e > 16$, no information can be secured via physical layer techniques.
This is because the eavesdropper can no longer be degraded by beamforming in this situation.

%

\section{Conclusion} \label{sec:con}
We have developed a  linear precoding scheme to achieve the capacity of
Gaussian multi-antenna wiretap channels in which the
legitimate receiver and eavesdropper have arbitrary numbers of antennas
but the transmitter has two antennas.
We have reformulated the  problem of determining the secrecy capacity into a tractable form and  solved this new  problem
to find the corresponding  optimal precoding and power allocation schemes.
Our investigation leads to a computable capacity with reasonably small complexity.
The gap between the secrecy rate achieved by the proposed precoding
and  GSVD-based beamforming can be remarkably high depending on the antenna configurations.
When the legitimate receiver or eavesdropper has  a single antenna, the optimal transmission scheme is unit-rank, i.e.,
beamforming is optimal.
  Further, in the absence of the eavesdropper, the proposed precoding
reduces to the capacity-achieving scheme of the MIMO/MISO channels. Hence,
 it can be used  for  these channels with/without an eavesdropper.

\section*{Appendix A:  Proof of Lemma~\ref{lem:lsump}}
\label{annexA}
\begin{proof}
Consider the optimization problem in \eqref{eq:cap0}. The secrecy capacity is zero if $\mathbf{H}^t\mathbf{H} - \mathbf{G}^t\mathbf{G} \preceq 0$ \cite{oggier2011secrecy}. In this case,  it is clear that $(\lambda_1,\lambda_2)=(0,0)$ is optimal.  Otherwise, the secrecy   capacity is strictly positive \cite{oggier2011secrecy} and  $\lambda_1+\lambda_2 = P$ is optimal. This completes the proof since the optimization problem in
   Lemma~\ref{lem:equi-opt} is a different representation of \eqref{eq:cap0}.
   \end{proof}

\section*{Appendix B}
\label{annexA}
To prove that the first (second) argument in \eqref{eq:thetaOPT1} corresponds to the
 minimum (maximum), it suffices to show that the second derivative of $W$ is positive for the first  argument and  negative for the second  one.
Let $\beta \triangleq \arcsin \frac{c}{\sqrt{a^2+b^2}}$ and recall that $\beta \in [-\frac{\pi}{2} \; \frac{\pi}{2}]$.
Then, form \eqref{eq:thetaOPT1}, the critical points are given by $\theta_1$ and  $\theta_2$ where
\begin{subequations}
\begin{align}
   \theta_1 &\triangleq   - \frac{1}{2} \phi - \frac{1}{2} \beta + k\pi,   \\
  \theta_2 &\triangleq  -  \frac{1}{2}\phi  + \frac{1}{2}\pi + \frac{1}{2}\beta  +k\pi.
\end{align}
\end{subequations}
Further, from  \eqref{eq:W}-\eqref{eq:rondsimp}, we know that
\begin{align}
\frac{\partial W}{\partial \theta} = \frac{\sin(2\theta + \phi) +\sin\beta  }{(a_2\sin2\theta + b_2 \cos2\theta + c_2)^2 }.
\end{align}
Then, at $\theta_2$  we have
\begin{align}
\frac{\partial^2 W}{\partial \theta^2}(\theta = \theta_2) &= \frac{2\cos(2\theta_2 + \phi)   }{(a_2\sin2\theta_2 + b_2 \cos2\theta_2 + c_2)^2 } \notag \\
& = \frac{-2\cos \beta    }{(a_2\sin2\theta_2 + b_2 \cos2\theta_2 + c_2)^2 }  \notag \\
& \le 0,
\end{align}
since $\beta \in [-\frac{\pi}{2} \; \frac{\pi}{2}]$.
Similarly, we can prove that $\frac{\partial^2 W}{\partial \theta^2}( \theta_1) \ge 0$.
Thus, $\theta_1$ and $\theta_2$  minimize and maximize  $W$, respectively.

\section*{Appendix C:  Proof of Lemma~\ref{lem:sym}}
To prove this, suppose $(\lambda_1 ,\lambda_2)$ maximizes  \eqref{eq:W}
 for  some $\theta^*$ given by \eqref{eq:thetaOPT}. Then, from \eqref{eq:num} and \eqref{eq:den},
 it is easy to check that  $(\lambda_2 ,\lambda_1)$ results in the same
$W$ for $\theta = \theta^* +\pi/2$. Therefore, $(\lambda_2 ,\lambda_1)$ can achieve the same rate as $(\lambda_1 ,\lambda_2)$ does. 

\section*{Acknowledgement}
The authors would like to thank the anonymous reviewers
for their valuable comments and suggestions that have significantly improved the
quality of the paper.


\begin{thebibliography}{10}

\bibitem{vaezi2017isit}
M.~Vaezi, W.~Shin, H.~V. Poor, and J.~Lee, ``{MIMO Gaussian wiretap channels
  with two transmit antennas: Optimal precoding and power allocation},'' in
  {\em Proc. IEEE International Symposium on Information Theory (ISIT)},
  pp.~1708--1712, 2017.

\bibitem{mukherjee2014principles}
A.~Mukherjee, S.~A.~A. Fakoorian, J.~Huang, and A.~L. Swindlehurst,
  ``Principles of physical layer security in multiuser wireless networks: A
  survey,'' {\em IEEE Communications Surveys \& Tutorials}, vol.~16, no.~3,
  pp.~1550--1573, 2014.

\bibitem{wyner1975wire}
A.~D. Wyner, ``The wire-tap channel,'' {\em The Bell System Technical Journal},
  vol.~54, no.~8, pp.~1355--1387, 1975.

\bibitem{khisti2010secure}
A.~Khisti and G.~W. Wornell, ``{Secure transmission with multiple
  antennas--Part II: The MIMOME wiretap channel},'' {\em IEEE Transactions on
  Information Theory}, vol.~56, no.~11, pp.~5515--5532, 2010.

\bibitem{oggier2011secrecy}
F.~Oggier and B.~Hassibi, ``{The secrecy capacity of the MIMO wiretap
  channel},'' {\em IEEE Transactions on Information Theory}, vol.~57, no.~8,
  pp.~4961--4972, 2011.

\bibitem{liu2009note}
T.~Liu and S.~Shamai, ``A note on the secrecy capacity of the multiple-antenna
  wiretap channel,'' {\em IEEE Transactions on Information Theory}, vol.~55,
  no.~6, pp.~2547--2553, 2009.

\bibitem{bustin2009mmse}
R.~Bustin, R.~Liu, H.~V. Poor, and S.~Shamai, ``{An MMSE approach to the
  secrecy capacity of the MIMO Gaussian wiretap channel},'' {\em EURASIP
  Journal on Wireless Communications and Networking}, no.~1, 2009.

\bibitem{li2013transmit}
Q.~Li, M.~Hong, H.-T. Wai, Y.-F. Liu, W.-K. Ma, and Z.-Q. Luo, ``{Transmit
  solutions for MIMO wiretap channels using alternating optimization},'' {\em
  IEEE Journal on Selected Areas in Communications}, vol.~31, no.~9,
  pp.~1714--1727, 2013.

\bibitem{steinwandt2014secrecy}
J.~Steinwandt, S.~A. Vorobyov, and M.~Haardt, ``{Secrecy rate maximization for
  MIMO Gaussian wiretap channels with multiple eavesdroppers via alternating
  matrix POTDC},'' in {\em Proc. IEEE International Conference on Acoustics,
  Speech and Signal Processing (ICASSP)}, pp.~5686--5690, 2014.

\bibitem{loyka2015algorithm}
S.~Loyka and C.~D. Charalambous, ``{An algorithm for global maximization of
  secrecy rates in Gaussian MIMO wiretap channels},'' {\em IEEE Transactions on
  Communications}, vol.~63, no.~6, pp.~2288--2299, 2015.

\bibitem{fakoorian2012optimal}
S.~A.~A. Fakoorian and A.~L. Swindlehurst, ``{Optimal power allocation for
  GSVD-based beamforming in the MIMO Gaussian wiretap channel},'' in {\em Proc.
  IEEE International Symposium on Information Theory}, pp.~2321--2325, 2012.

\bibitem{shafiee2009towards}
S.~Shafiee, N.~Liu, and S.~Ulukus, ``{Towards the secrecy capacity of the
  Gaussian MIMO wire-tap channel: The 2-2-1 channel},'' {\em IEEE Transactions
  on Information Theory}, vol.~55, no.~9, pp.~4033--4039, 2009.

\bibitem{shafiee2007achievable}
S.~Shafiee and S.~Ulukus, ``{Achievable rates in Gaussian MISO channels with
  secrecy constraints},'' in {\em Proc. IEEE International Symposium on
  Information Theory}, pp.~2466--2470, 2007.

\bibitem{khisti2010secureMISOME}
A.~Khisti and G.~W. Wornell, ``{Secure transmission with multiple antennas I:
  The MISOME wiretap channel},'' {\em IEEE Transactions on Information Theory},
  vol.~56, no.~7, pp.~3088--3104, 2010.

\bibitem{loyka2012optimal}
S.~Loyka and C.~D. Charalambous, ``{On optimal signaling over secure MIMO
  channels},'' in {\em Proc. IEEE International Symposium on Information Theory
  (ISIT)}, pp.~443--447, 2012.

\bibitem{fakoorian2013full}
S.~A.~A. Fakoorian and A.~L. Swindlehurst, ``{Full rank solutions for the MIMO
  Gaussian wiretap channel with an average power constraint},'' {\em IEEE
  Transactions on Signal Processing (ISIT)}, vol.~61, no.~10, pp.~2620--2631,
  2013.

\bibitem{loyka2016optimal}
S.~Loyka and C.~D. Charalambous, ``{Optimal signaling for secure communications
  over Gaussian MIMO wiretap channels},'' {\em IEEE Transactions on Information
  Theory}, vol.~62, no.~12, pp.~7207--7215, 2016.

\bibitem{zhang2014energy}
H.~Zhang, Y.~Huang, S.~Li, and L.~Yang, ``{Energy-efficient precoder design for
  MIMO wiretap channels},'' {\em IEEE Communications Letters}, vol.~18, no.~9,
  pp.~1559--1562, 2014.

\bibitem{shlezinger2017secrecy}
N.~Shlezinger, D.~Zahavi, Y.~Murin, and R.~Dabora, ``The secrecy capacity of
  {Gaussian MIMO} channels with finite memory,'' {\em IEEE Transactions on
  Information Theory}, vol.~63, no.~3, pp.~1874--1897, 2017.

\bibitem{wang2014joint}
H.-M. Wang, F.~Liu, and X.-G. Xia, ``{Joint source-relay precoding and power
  allocation for secure amplify-and-forward MIMO relay networks},'' {\em IEEE
  Transactions on Information Forensics and Security}, vol.~9, no.~8,
  pp.~1240--1250, 2014.

\bibitem{fang2015game}
B.~Fang, Z.~Qian, W.~Shao, W.~Zhong, and T.~Yin, ``Game-theoretic precoding for
  cooperative {MIMO SWIPT} systems with secrecy consideration,'' in {\em Proc.
  IEEE Global Communications Conference (GLOBECOM)}, pp.~1--5, 2015.

\bibitem{he2014mimo}
X.~He and A.~Yener, ``{MIMO} wiretap channels with unknown and varying
  eavesdropper channel states,'' {\em IEEE Transactions on Information Theory},
  vol.~60, no.~11, pp.~6844--6869, 2014.

\bibitem{tse2005fundamentals}
D.~Tse and P.~Viswanath, {\em Fundamentals of Wireless Communication}.
\newblock Cambridge University Press, 2005.

\bibitem{csiszar1978broadcast}
I.~Csisz{\'a}r and J.~Korner, ``Broadcast channels with confidential
  messages,'' {\em IEEE Transactions on Information Theory}, vol.~24, no.~3,
  pp.~339--348, 1978.

\bibitem{li2007secret}
Z.~Li, W.~Trappe, and R.~Yates, ``Secret communication via multi-antenna
  transmission,'' in {\em Proc. 41st Annual Conference on Information Sciences
  and Systems (CISS)}, pp.~905--910, 2007.

\bibitem{bashar2012secrecy}
S.~Bashar, Z.~Ding, and C.~Xiao, ``{On secrecy rate analysis of MIMO wiretap
  channels driven by finite-alphabet input},'' {\em IEEE Transactions on
  Communications}, vol.~60, no.~12, pp.~3816--3825, 2012.

\bibitem{li2009transmitter}
J.~Li and A.~Petropulu, ``{Transmitter optimization for achieving secrecy
  capacity in Gaussian MIMO wiretap channels},'' {\em arXiv preprint
  arXiv:0909.2622}, 2009.

\bibitem{khina2014ordinary}
A.~Khina, Y.~Kochman, and A.~Khisti, ``{From ordinary AWGN codes to optimal
  MIMO wiretap schemes},'' in {\em Proc. IEEE Information Theory Workshop
  (ITW)}, pp.~631--635, 2014.

\bibitem{tyagi2014explicit}
H.~Tyagi and A.~Vardy, ``{Explicit capacity-achieving coding scheme for the
  Gaussian wiretap channel},'' in {\em Proc. IEEE International Symposium on
  Information Theory}, pp.~956--960, 2014.

\end{thebibliography}
\end{document}